\documentclass{birkjour}

\usepackage[noadjust]{cite}
\usepackage{xcolor}
\RequirePackage[all]{xy}

\usepackage{amsmath,amssymb,amsthm}
\usepackage{mathrsfs,mathtools}

\usepackage{graphicx}
\usepackage{float}
\usepackage{indentfirst}

\newtheorem{theorem}{Theorem}[section]

\newtheorem{lemma}[theorem]{Lemma}
\newtheorem{proposition}[theorem]{Proposition}
\theoremstyle{definition}
\newtheorem{definition}[theorem]{Definition}

\begin{document}
	\title[Intuitionistic Quantum Logic Perspective]{Intuitionistic Quantum Logic Perspective: Static and Dynamic Revision Operators}
	
	\author[Heng]{Heng Zhou}
	\address{School of Mathematical Science, Beihang University, 37 Xueyuan Road, Haidian District, Beijing, 100191, China}
	\email{zhouheng@buaa.edu.cn}
	
	\thanks{This work was supported by the National Natural Science Foundation of China (Grant No. 12371016, 11871083).}
	
	\author[Yongjun]{Yongjun Wang}
	\address{School of Mathematical Science, Beihang University, 37 Xueyuan Road, Haidian District, Beijing, 100191, China\\
	\textbf{Corresponding author}}
	\email{wangyj@buaa.edu.cn}
	
	\author[Baoshan]{Baoshan Wang}
	\address{School of Mathematical Science, Beihang University, 37 Xueyuan Road, Haidian District, Beijing, 100191, China}
	\email{bwang@buaa.edu.cn}
	
	\author[Jian]{Jian Yan}
	\address{School of Mathematical Science, Beihang University, 37 Xueyuan Road, Haidian District, Beijing, 100191, China}
	\email{jianyanmath@buaa.edu.cn}
	
	\author[Xiaoyang]{Xiaoyang Wang}
	\address{School of Philosophy, Beijing Normal University, Beijing, China}
	\email{wangxiaoyang@bnu.edu.cn}
	
	
	\date{\today}
	
	\begin{abstract}
		The classical belief revision framework, as proposed by Alchourron, Gardenfors, and Makinson, involves the revision of a theory based on eight postulates. In this paper, we focus on the exploration of a revision theory grounded in quantum mechanics, referred to as the natural revision theory. 
		
		There are two reasoning modes in quantum systems: static intuitionistic reasoning, which incorporates contextuality, and dynamic reasoning, which is achieved through projection measurement. We combine the advantages of two intuitionistic quantum logic frameworks, as proposed by D{\"o}ring and Coecke, respectively. Our goal is to establish a truth-value assignment for intuitionistic quantum logic that not only aligns with the inherent characteristics of quantum mechanics but also supports truth-value reasoning. The natural revision theory is then investigated based on this approach.
		
		We introduce two types of revision operators that correspond to the two reasoning modes in quantum systems: static and dynamic revision. Furthermore, we highlight the distinctions between these two operators. Shifting away from classical revision paradigms, we consider the revision of consequence relations in intuitionistic quantum logic. We demonstrate how, within the natural revision theory framework, both revision operators collectively influence the consequence relations. Notably, the outcomes of revision process are impacted by the sequence in which these interweaved operators are deployed.
	\end{abstract}

	\maketitle
	
	
	\section{Introduction}
	
	Classical belief revision was first proposed by Alchourron, Gardenfors, and Makinson in 1985\cite{1985On}. A theory is contracted or revised in terms of six elementary and two supplementary postulates. Their research centered on the methodology for selecting subsets from an original theory that align with newly acquired message, a process typically influenced by subjective criteria. In this paper, our primary focus is on the exploration of a revision theory known as ``natural revision," which is grounded in the fundamental principles of quantum mechanics. This methodology enables a selection process governed by objective criteria, thereby distinguishing it from methods reliant on subjective judgment.
	
	The logical characterization of quantum mechanics was introduced by Birkhoff and von Neumann in 1936\cite{birkhoff1936logic}. They showed that a quantum system can be described via Hilbert space. There is a one-to-one correspondence between projection operators of a quantum system and closed subspaces of Hilbert space. The operations between closed subspaces satisfy orthomodular law. Quantum logic, represented by an orthomodular lattice composed of closed subspaces, is referred to as standard quantum logic. It is still a hot spot in the field of quantum computing\cite{saharia2019elementary,babu2023quantum}.
	
	The reasoning of quantum logic is non-monotonic reasoning based on features of quantum mechanics. Monotonicity means that new message added to the reasoning does not affect the result of the implication of old messages. In a quantum system, when reasoning is revised by projection measurement as new message, some results implied from the old ones are no longer valid because the projection measurement changes the state of the system. We intend to investigate the nature of reasoning in quantum logic and thus provide a natural revision theory determined by the fundamental laws of quantum mechanics.
	
	Standard quantum logic takes closed subspaces of Hilbert space as elements, the intersection of closed subspaces as the meet operation, the span of closed subspaces as the join operation, and the orthogonal complement of a closed subspace as the negation operation, which forms an orthomodular lattice called the property lattice. There is no implication connective in the property lattice, so truth-value reasoning cannot be carried out. 
	
	Finch\cite{finch_1969,1970Quantum} defined the Sasaki hook $\rightsquigarrow_s$ and it is used as the implication connective in property lattice. The left adjoint of Sasaki hook, called Sasaki projection $\wedge_s$, characterizes projection measurements in a quantum system. Roman and Zuazua\cite{roman1999quantum} discussed that the Sasaki adjoint is some kind of deduction theorem. Engesser and Gabbay\cite{2002Quantum} inspected the link between non-monotonic consequence relations and belief revision. They argued that reasoning in quantum logic is a non-monotonic consequence relation, and the Sasaki projection is the revision operation of consequence relations. The adjoint operation Sasaki hook is the internalizing connective. A consequence revision system (CRS) is constructed by using formulas as operators. They examined the consequence relations based on truth-value reasoning in the property lattice, and regard $a,b,c$ in $a\wedge_s b\leq c$ as arguments. With $\wedge_s$ as the revision operation, formulas $a$ and $b$ can be used as operators to revise a consequence relation. Sasaki adjoint can operate in the form of modus ponens, i.e., $a\wedge_s(a\rightsquigarrow_s b)\leq b$. However, because of the different meanings of elements in the operation, Sasaki adjoint cannot play the role of truth-value reasoning. 
	
	Bob Coecke\cite{coecke2004sasaki} indicated that Sasaki hook is not a static implication connective but induces a dynamic one with a parameter in dynamic operational quantum logic. Define $\varphi_a^*(b):=a\wedge_s b$ as well as $\varphi_{a,*}(b):=a\rightsquigarrow_s b$, and $\varphi_a^* (b)\leq c \Leftrightarrow b\leq \varphi_{a,*}(c)\Leftrightarrow b\stackrel{\varphi_a}{\rightsquigarrow} c$ is satisfied, in which $b$ and $c$ are variables while $a$ is a parameter. Obviously, the meaning of $a$ in the operation differs from that of $b$ and $c$. $\varphi_a^* (b)\leq c$ represents that with space $a$ as the projection space, the result of projecting $b$ space to $a$ is contained in $c$ space. In \cite{2002Quantum}, considering only the revision form based on the projection operation is not enough to characterize quantum reasoning. We investigate the natural revision theory based on quantum logic by exploring both static and dynamic aspects of truth-value reasoning in quantum logic.
	
	In our study, we explore the utilization of topos theory to depict the contextuality of properties within quantum systems. Isham et al.\cite{isham1998topos,hamilton2000topos} characterized physical quantities of a quantum system by presheaf structure. They emphasized that contextuality is a crucial characteristic of quantum systems, and properties in a quantum system should be assigned local truth values within specific contexts. D{\"o}ring et al.\cite{doring2008topos1,doring2008topos2} defined a mapping called daseinisation to approximate a property across all contexts. They suggested that the global truth value of properties can be determined based on the outer daseinisation. D{\"o}ring\cite{doring2011topos,doring2016topos} showed that elements corresponding to the outer daseinisations of properties generate a bi-Heyting algebra. 
	
	However, upon investigating the implication connective in this bi-Heyting algebra, we find that the implication connective in the sense of Heyting algebra is trivial,	while the one in the sense of co-Heyting algebra is non-trivial. Therefore, the elements selected by D{\"o}ring constitute a co-Heyting algebra that satisfies the law of excluded middle, whereas ordinary truth-value reasoning should rely on a Heyting algebra that satisfies the law of non-contradiction as the reasoning structure. 
	
	Coecke\cite{coecke2002quantum} constructed an intuitionistic quantum logic based on Heyting algebra using the downsets of elements in the property lattice. However, while Coecke's structure aligns with conventional truth-value reasoning, it does not involve contextuality, and is insufficient to describe the characteristics of quantum mechanics. In their research on the dynamic aspect of quantum logic\cite{coecke2004logic}, Coecke et al. highlighted that a causal relationship based on actions can be induced to replace the partial order relationship between elements. They defined a pair of adjoint according to this causal relationship, where the adjoint operations respectively represent the propagation and causation of properties. Reasoning in dynamic operational quantum logic is then conducted based on this adjoint relationship. We intend to propose an intuitionistic quantum logic that combines the respective advantages of static and dynamic quantum logics mentioned above. This logic satisfies the characteristics of quantum mechanics and allows for truth-value reasoning.
	
	In our approach, we propose selecting elements according to the inner daseinisation for the global truth value of properties. This selection guarantees that the chosen elements form a Heyting algebra. Truth-value reasoning within a Heyting algebra is static reasoning. However, projection measurement, serving as an external action, transforms one Heyting algebra into another, and the truth-value assignment generated by the Sasaki projection represents dynamic truth-value changes between Heyting algebras.
	
	We believe that a consequence relation occurs within a Heyting algebra. There are two types of revision operations for a consequence relation. The first is to add new assumptions into the Heyting algebra as new message, corresponding to the $\wedge$ operation as revision, which we refer to as static revision. The second is the action of projection measurement between Heyting algebras as new message, corresponding to the $\wedge_s$ operation as revision, which we refer to as dynamic revision. Both types of revision use formulas as operators to revise a consequence relation. The distinction lies in the methodology employed to revise the antecedents of a consequence relation. Static revision changes the consequence relation within a single Heyting algebra, resulting in a new consequence relation that remains within the original Heyting algebra. Meanwhile, dynamic revision involves altering the consequence relation from one within a given Heyting algebra to a new consequence relation situated within a different Heyting algebra.
	
	Engesser et al. utilized the orthomodular lattice to represent the natural revision theory. However, the orthomodular lattice fails to differentiate between the static and dynamic operations, resulting in a single reasoning mode. Consequently, their work \cite{2002Quantum} only includes one revision operator, which corresponds to the ``expansion" in the revision theory. In contrast, intuitionistic quantum logic allows for the distinction between static reasoning and dynamic one. In our approach, we propose two kinds of revision operators that correspond to the ``expansion" and ``revision" in the revision theory, respectively. The outcomes of revision process are influenced by the order in which the interweaved operators are applied. These operators enable the description of the natural revision theory following the fundamental characteristics of quantum mechanics.
	
	In section 2, we provide the basic knowledge and background relevant to this paper. In section 3, we review the existing revision theory, discuss the revision of consequence relation based on the operations in the orthomodular lattice by Engesser et al., and highlight the shortcomings of this method. We then investigate the truth assignment of properties in a quantum system based on topos theory, and present a more suitable method for truth-value reasoning in section 4. In section 5, we discuss the static and dynamic reasoning of quantum logic based on the provided truth assignment, consider the two types of revision operators for a consequence relation, and present the natural revision theory based on the truth-value reasoning of quantum logic. Finally, we summarize the paper and propose prospects for research.
	
	\section{Preliminaries}
	\subsection{Standard quantum logic}
	A quantum physical system is represented by a Hilbert space $\mathcal{H}$ over complex field $\mathbb{C}$, where a physical quantity is represented by a self-adjoint operator $\hat{A}$ in $\mathcal{H}$. The value of a physical quantity $\hat{A}$ is taken in a Borel set $\Delta$, called $\hat{A}\in\Delta$. A projective operator(or projector in short) $\hat{P}$ represents the event of assigning values to a physical quantity, which is denoted by $\hat{P}=E[\hat{A}\in \Delta]$. Projectors are special self-adjoint operators. There is a one-to-one correspondence between projectors and closed subspaces of a Hilbert space. A closed subspace corresponds to a property of a quantum system whereas a one-dimension closed subspace corresponds to a state of the system. We mainly investigate projectors while studying quantum systems. In other words, we do not care what value a physical quantity gets, but whether an event that a quantity gets a certain value is true.
	
	\begin{definition}
		Define $\langle L,\leq,\wedge,\vee,\neg\rangle$ be an ortholattice if $\langle L,\leq,\wedge,\vee \rangle$ is bounded, and a unary operation $\neg$ defined in the lattice satisfies for any $a,b\in L$, there are:
		\begin{itemize}
			\item $\neg(\neg a)=a$
			\item $a\leq b \Leftrightarrow \neg b\leq \neg a$
			\item $a\wedge \neg a=0$
			\item $a\vee \neg a=1$
		\end{itemize}
	\end{definition}
	We make no distinction between lattice $\langle L,\leq,\wedge,\vee \rangle$ and set $L$ of elements of the lattice.
	\begin{definition}
		An ortholattice $L$ is called an orthomodular lattice if, for any $a,b\in L,a\leq b$, the Orthomodular Law is satisfied.
		\begin{align*}
			a\vee (\neg a \wedge b)=b
		\end{align*}
	\end{definition}
	
	Given a Hilbert space $\mathcal{H}$, denote a collection of its closed subspaces $Sub(\mathcal{H})$. Define partial order relation $\leq$ to be the inclusion of closed subspaces. $\wedge$ operation is the intersection of two closed subspaces.  $\vee$ is the span of two closed subspaces. $\neg$ is the complement of a closed subspace. $\langle Sub(\mathcal{H}),\leq,\wedge,\vee,\neg\rangle$ is an orthomodular lattice. There is a one-to-one correspondence between closed subspaces of a Hilbert space and projectors in a quantum system, and a projector represents a property in the system. Hence the orthomodular lattice is called a property lattice. Without causing ambiguity, we denote the property lattice by $L$.
	
	\begin{definition}
		Given a pair of mappings $f:M\rightarrow N$ and $g:N\rightarrow M$. Define $f$ and $g$ be a pair of Galois adjoint, denoted by $f\dashv g$, if $f(a)\leq b\Leftrightarrow a\leq g(b)$ is satisfied for any $a\in M,b\in N$. Call $f$ a left adjoint of $g$, and $g$ a right adjoint of $f$.
	\end{definition}
	
	A pair of operations $\wedge_s$ and $\rightsquigarrow_s$ are introduced into the property lattice to describe the characteristics of quantum mechanics. Furthermore, the mappings $a\wedge_s \_$ and $a\rightsquigarrow_s\_$ form a pair of adjoints as well, called Sasaki adjoint, which is denoted by $a\wedge_s\_\dashv a\rightsquigarrow_s\_$. In which
	\begin{align*}
		a\wedge_s b:=a\wedge(\neg a\vee b)
	\end{align*}
	expresses the projection of $b$ onto $a$, or projecting $b$ by $a$. We notice that there is no commutative law for $\wedge_s$. The right adjoint
	\begin{align*}
		a\rightsquigarrow_s b:=\neg a\vee (a\wedge b)
	\end{align*}
	reflects the reasoning ability of the property lattice. The relation between Sasaki adjoint
	\begin{align*}
		a\wedge_s b\leq c\Leftrightarrow b\leq a\rightsquigarrow_s c
	\end{align*}
	is analogous to modus ponens in formal systems and to obtain
	\begin{align*}
		a\wedge_s (a\rightsquigarrow_s c)\leq c
	\end{align*}

	\subsection{Topos quantum theory}
	An alternative mathematical characterization of quantum logic is given by Isham and Butterfield et al.\cite{isham1998topos,hamilton2000topos} using topos structure $Set^{\mathcal{V}(\mathcal{H})^{op}}$. They emphasize the crucial role of contextuality in depicting quantum logic. They indicate that self-adjoint operators in a Hilbert space $\mathcal{H}$ constitute a von Neumann algebra $N(\mathcal{H})$, and every commutative von Neumann subalgebra $V$ in $N(\mathcal{H})$ represents a context. The collection of all commutative von Neumann subalgebras is denoted by $\mathcal{V}(N(\mathcal{H}))$, and by $\mathcal{V}(\mathcal{H})$ for short. For more knowledge about topos quantum theory, please refer to \cite{flori2012lectures}.
	\begin{definition}
		A context is a commutative von Neumann subalgebra. For a $n$-dimensional context $V$, the set of generators of $V$ is denoted by $\mathcal{F}_{V}=\{\hat{P_1},\hat{P_2},...,\hat{P_n}\}$. Every component of the spectral decomposition of self-adjoint operators in $V$ is an element in $\mathcal{F}_{V}$. Every self-adjoint operator in $V$ is a linear combination of the generators above. If we only consider all projectors in $V$, which is denoted by $P(V)$, then $P(V)$ constitutes a Boolean subalgebra.
	\end{definition}
	\begin{definition}
		Given a $n$-dimensional commutative von Neumann algebra $V$, define the set $\sigma(V)=\{\lambda_1,\lambda_2,...,\lambda_n\}$ to be the spectrum of $V$, where every $\lambda_i$ maps a self-adjoint operator in $V$ to the coefficient of $\hat{P_i}$ component of spectral decomposition of the operator. $\lambda_i$ is of one-to-one correspondence to $\hat{P_i}$.
	\end{definition}
	\begin{definition}
		The spectral presheaf in $Set^{\mathcal{V}(\mathcal{H})^{op}}$ is a contravariant functor $\underline{\Sigma}\colon \mathcal{V}(\mathcal{H})^{op}\rightarrow Set$, which
		\begin{itemize}
			\item for objects: $\underline{\Sigma}(V) \coloneq \sigma(V)$
			\item for morphisms: if $V_2 \subseteq V_1$, i.e. there is a morphism $i_{V_2 V_1}\colon V_2 \rightarrow V_1$ in $\mathcal{V}(\mathcal{H})$, then $\underline{\Sigma}(i_{V_2 V_1})\colon \sigma(V_1)\rightarrow \sigma(V_2)$ is a restriction mapping.
		\end{itemize}
	\end{definition}
	\begin{definition}
		For any context $V$, define an isomorphism between projectors and subsets of the spectral presheaf by
		\begin{align*}
			\alpha_{V}\colon P(V) &\rightarrow Cl(\underline{\Sigma}_{V})\\
			\hat{P} &\mapsto \{\lambda\in\underline{\Sigma}_{V}| \lambda(\hat{P})=1\}
		\end{align*}
		in which $Cl(\underline{\Sigma}_{V})$ represents clopen subsets of $\underline{\Sigma}_{V}$.
	\end{definition}
	
	A property does not belong to every context in general. D{\"o}ring\cite{doring2008topos1,doring2008topos2} indicates that for the projector corresponding to a property, one can find an approximation of the projector in every context. 
	\begin{definition}
		Given a projector $\hat{P}$ corresponding to a property, define two approximations of $\hat{P}$ in context $V$ by
		\begin{itemize}
			\item $\delta^{o}(\hat{P})_V = \bigwedge \{\hat{Q}\in P(V)|\hat{P}\leq\hat{Q} \}$ and
			\item $\delta^{i}(\hat{P})_V = \bigvee \{\hat{Q}\in P(V)|\hat{Q}\leq\hat{P} \}$
		\end{itemize}
		where $\delta^o$ is called outer daseinisation, which means approximate from above; and $\delta^i$ is called inner daseinisation, which means approximate from below.
	\end{definition}
	
	It should be noted that D{\"o}ring employs only the outer daseinisation for approximating a projector. Without causing confusion, D{\"o}ring refers to outer daseinisation simply as daseinisation. However, in this paper, it is necessary to make a clear distinction between outer and inner daseinisation.
	
	\subsection{Consequence relation}
	According to \cite{2002Quantum}, given a class $Fml$ of formulas closed under the connectives $\neg,\wedge,\vee$ and containing $\top$ and $\bot$ for truth and falsity respectively. A consequence relation $\vdash\subseteq Fml\times Fml$ should satisfy the following conditions.
	\begin{itemize}
		\item Reflexivity
		\begin{align*}
			\alpha\vdash\alpha
		\end{align*}
		\item Cut
		\begin{align*}
			\frac{\alpha\wedge\beta\vdash\gamma , \alpha\vdash\beta}{\alpha\vdash\gamma}
		\end{align*}
		\item Restricted Monotonicity
		\begin{align*}
			\frac{\alpha\vdash\beta , \alpha\vdash\gamma}{\alpha\wedge\beta\vdash\gamma}
		\end{align*}
	\end{itemize}
	
	The relation between general non-monotonic consequence relations and revision operators is as follows. Let $\Delta$ be a theory, $\circ$ a revision operation. Define $\vdash_\Delta$ by
	\begin{align*}
		A\vdash_\Delta B \qquad \text{iff} \qquad B\in\Delta\circ A.
	\end{align*}
	Given a $\vdash$, $\circ$ is obtained by
	\begin{align*}
		\Delta\circ A=\{B|\Delta,A\vdash B\}
	\end{align*}
	
	\section{Revision theory in standard quantum logic}
	In their work on standard quantum logic, Engesser et al.\cite{2002Quantum} propose a revision theory based on the Sasaki adjoint operation of property lattice, treating all elements within this operation uniformly. However, Bob Coecke\cite{coecke2004sasaki} has demonstrated that the roles of elements in Sasaki adjoint are distinct. This discrepancy suggests that the revision theory defined by standard quantum logic may not fully align with the fundamental law of quantum mechanics, nor accurately capture its characteristics.
	
	\begin{definition}
		Let $\vdash$ be a consequence relation. Define a internalizing connective $\rightsquigarrow$ of a consequence relation $\vdash$ as follows: For any formulas $\alpha$ and $\beta$, $\alpha\vdash \beta$ if and only if $\vdash\alpha\rightsquigarrow\beta$.
	\end{definition}
	
	It is worth noting that in classical logic, the material implication $\rightarrow$ is the internalizing connective of classical consequence relation, and the definition above corresponds to the deduction theorem of classical logic.
	
	By analogy with operations in classical logic, Engesser et al. try to give a consequence relation and internalizing connective satisfying quantum logic reasoning. Inspired by classical logic, treating formulas as operators, they define a consequence revision system (CRS).
	
	\begin{definition}
		Let $Fml$ be a class of formulas, and $\mathcal{C}$ be a class of consequence relations on $Fml\times Fml$. Let function $F$ be $F:Fml\times\mathcal{C}\rightarrow\mathcal{C}$. Call $F$ an action on $\mathcal{C}$, if for any consistent consequence relations $\vdash\in\mathcal{C}$ and $\alpha,\beta\in Fml$ the following conditions are satisfied.
		\begin{itemize}
			\item $F(\top,\vdash)=\vdash$
			\item $F(\alpha,\vdash)=0 \quad \text{iff} \quad \vdash \neg \alpha$
			\item $F(\beta,F(\alpha,\vdash))=F(\alpha,\vdash) \quad \text{iff} \quad \alpha\vdash\beta$.
		\end{itemize}
	
		The term ``consistent consequence relation" refers to a consequence relation that does not yield the contradiction represented by the formula ``0".
		
		If $F$ is an action on $\mathcal{C}$, then call $\langle \mathcal{C},F\rangle$ a consequence revision system CRS.\textsuperscript{\cite{2002Quantum}}
	\end{definition}
	
	Every formula $\alpha\in Fml$ corresponds to a revision operator $\bar{\alpha}:\mathcal{C}\rightarrow\mathcal{C}$ on consequence relations, which is defined by $\bar{\alpha}\vdash:=F(\alpha,\vdash)$. We denote $\bar{\alpha}\vdash$ by $\vdash_\alpha$.
	
	Treat an action of a formula to a consequence relation as the formula operates on the formulas in the hypothesis set of the consequence relation. Define $*$ be the revision operation on the set of formulas if the following conditions are satisfied. 
	\begin{itemize}
		\item $1*a=a*1=a$
		\item there exists an adjoint $a*\_ \dashv a\rightsquigarrow \_$
		\item $a*b\leq a$
		\item $a\bot b:=a\leq \neg b$ \quad iff \quad $a*b=0$
		\item $b\leq a$ \quad iff \quad $a*b=b$
	\end{itemize}
	Then, the action of function $F:Fml\times \mathcal{C}\rightarrow\mathcal{C}$ on a consequence relation is $F(a,\vdash_b)=\vdash_{a*b}$.
	
	Regarding the operation $\wedge_s$ in property lattice $L$ as the revision operation, partial relation on property lattice as consequence relation, the connective $\rightsquigarrow_s$ as the internalizing connective of consequence relation, it is easy to verify that $\wedge_s$ satisfies the conditions of revision operation, and for any $x\in L$, $\vdash_x$ satisfies the conditions of consequence relation.
	
	Now in property lattice $L$, there is
	\begin{align*}
		a\wedge_s b\vdash c \Leftrightarrow a\vdash_b c \Leftrightarrow \vdash_{a\wedge_s b}c \Leftrightarrow\vdash_b a\rightsquigarrow_s c \Leftrightarrow b\vdash a\rightsquigarrow_s c
	\end{align*}
	
	Notice that for $a\wedge_s b\vdash c\Leftrightarrow a\vdash_b c$, the consequence relation $\vdash$ becomes the consequence relation $\vdash_b$ when revised by formula $b$. That is $F(b,\vdash)=\vdash_b$. By the definition of revision operation, let $\vdash=\vdash_1$, there is $F(b,\vdash)=F(b,\vdash_1)=\vdash_{b\wedge_s 1}$. Furthermore, for $a\wedge_s b\vdash c \Leftrightarrow a\vdash_b c \Leftrightarrow \vdash_{a\wedge_s b} c$, there is $F(a,F(b,\vdash_1))=F(a,\vdash_b)=\vdash_{a\wedge_s b}$. Firstly, $b$ revises the consequence $\vdash = \vdash_1$, and then $a$ revises the result from the previous step. 
	
	This approach is to consider the operations of property lattice as purely mathematical, which ignores the physical meaning of the operations. For equation $a\wedge_s b = a\wedge(\neg a\vee b)=c$, this approach treats the variables $a,b,c$ as indistinctive variables. However, when considering the physical interpretation of operation $\wedge_s$, the equation $a\wedge_s b=c$ should be interpreted as taking space $a$ as a projection space, projecting space $b$ onto space $a$, and obtaining space $c$ as the result. Moreover, Bob Coecke showed that the Sasaki hook is not a static implication connective but induces a dynamic one with a parameter. Define $\varphi_a^*(b):=a\wedge_s b$ and $\varphi_{a,*}(b):=a\rightsquigarrow_s b$, and then $\varphi_a^* (b)\leq c \Leftrightarrow b\leq \varphi_{a,*}(c)\Leftrightarrow b\stackrel{\varphi_a}{\rightsquigarrow} c$ is obtained, where $b$ and $c$ are variables but $a$ is a parameter, which means they are elements from different levels. As external operations of property lattice $L$, the projection operation and its adjoint do not carry out inside a property lattice, but between two property lattices $L_1$ and $L_2$. Specifically, the element $b$ assigned true in property lattice $L_1$ is acted upon by a projection operation $\varphi_a$ corresponding to space $a$, resulting in the element $c$ in property lattice $L_2$ being assigned true. Therefore, the consequence relation system(CRS) with internalizing connective based on standard quantum logic may not reflect the features of truth-value operations in quantum logic. 
	
	The revision theory based on standard quantum logic may not accurately reflect the nature of quantum mechanics. To address this problem, we attempt to define an alternative natural revision theory conforming to the characteristics of quantum mechanics. To achieve this, the logical formalism of quantum systems and the truth-value assignment of quantum logic need to be reconsidered. 
	
	\section{Truth assignment in topos quantum theory}
	According to the Kochen-Specker theorem\cite{kochen1990problem}, all physical quantities cannot be assigned simultaneously and satisfy the functional relation between physical quantities. Isham et al. indicate that contextuality is the core concept to depict a quantum system. Investigating the assignment of physical quantities should examine the assignment of quantities in all contexts. Hence Isham et al. intend to characterize quantum systems with topos theory and inspect the truth-value assignment therein.
	
	\subsection{Bi-Heyting algebra structure}
	A subobject $\underline{S}$ of $\underline{\Sigma}$ such that the components $\underline{S}_{V}$ are clopen sets for all $V$ is called a clopen subobject.
	
	\begin{definition}
		Let $N(\mathcal{H})$ be the von Neumann algebra constituted by self-adjoint operators in Hilbert space $\mathcal{H}$. Let $L$ be the lattice constituted by projection operators in the von Neumann algebra. The mapping
		\begin{align*}
			\underline{\delta}^{o}\colon L&\rightarrow Sub_{cl}(\underline{\Sigma})\\
			\hat{P}&\mapsto \underline{\delta^{o}(\hat{P})}\coloneq (\alpha_{V}(\delta^{o}(\hat{P})_V))_{V\in N(\mathcal{H})}
		\end{align*}
		is called the outer daseinisation of projections. Here, $\underline{\delta^{o}(\hat{P})}\in Sub_{cl}(\underline{\Sigma})$ represents the subobject of the spectral presheaf $\underline{\Sigma}$.
	\end{definition}
	D{\"o}ring shows that the mapping $\underline{\delta}^{o}$ preserves all joins and is an order-preserving injection, but not a surjection.
	It is evident that $\underline{\delta^{o}(\hat{0})} = \underline{0}$, the empty subobject, and $\underline{\delta^{o}(\hat{1})} = \underline{\Sigma}$. For joins, we have
	\begin{align*}
		\forall \hat{P},\hat{Q}\in L\colon \underline{\delta^{o}(\hat{P}\vee \hat{Q})}= \underline{\delta^{o}(\hat{P})}\vee\underline{\delta^{o}(\hat{Q})}.
	\end{align*}
	However, for meets, we have
	\begin{align*}
		\forall \hat{P},\hat{Q}\in L\colon \underline{\delta^{o}(\hat{P}\wedge \hat{Q})}\leq \underline{\delta^{o}(\hat{P})}\wedge\underline{\delta^{o}(\hat{Q})}.
	\end{align*}
	In general, $\underline{\delta^{o}(\hat{P})}\wedge\underline{\delta^o (\hat{Q})}$ is not of the form $\underline{\delta^{o}(\hat{R})}$ for any projection $\hat{R}\in L$.
	
	D{\"o}ring argues that the elements of $Sub_{cl}(\underline{\Sigma})$ form a complete bi-Heyting algebra. Let $\underline{S},\underline{T}\in Sub_{cl}(\underline{\Sigma})$ be two clopen subobjects. We have the following equations:
	\begin{align*}
		\forall V\in \mathcal{V}(\mathcal{H})\colon (\underline{S}\wedge\underline{T})_V &= \underline{S}_V \cap \underline{T}_V \\
		(\underline{S}\vee\underline{T})_V &= \underline{S}_V \cup \underline{T}_V.
	\end{align*}
	
	For each $\underline{S}\in Sub_{cl}(\underline{\Sigma})$, the functor
	\begin{align*}
		\underline{S}\wedge \_\colon Sub_{cl}(\underline{\Sigma}) \rightarrow Sub_{cl}(\underline{\Sigma})
	\end{align*}
	has a right adjoint given by
	\begin{align*}
		\underline{S}\rightarrow \_ \colon Sub_{cl}(\underline{\Sigma}) \rightarrow Sub_{cl}(\underline{\Sigma}).
	\end{align*}
	The Heyting implication is determined by the adjunction:
	\begin{align*}
		\underline{S}\wedge \underline{R}\leq \underline{T} \quad \text{iff} \quad \underline{R} \leq (\underline{S}\rightarrow\underline{T}).
	\end{align*}
	This implies that
	\begin{align*}
		(\underline{S}\rightarrow\underline{T})=\bigvee\{\underline{R}\in Sub_{cl}(\underline{\Sigma})| \underline{S}\wedge\underline{R}\leq\underline{T}\}.
	\end{align*}
	The contextwise definition is as follows: for every $V\in \mathcal{V}(\mathcal{H})$,
	\begin{align*}
		(\underline{S}\rightarrow\underline{T})_V = \{\lambda\in \underline{\Sigma}_{V}| \forall V' \subseteq V \colon \text{if} \ \lambda|_{V'}\in \underline{S}_{V'}, \text{then} \ \lambda |_{V'}\in \underline{T}_{V'}\}.
	\end{align*}
	The Heyting negation $\neg$ for each $\underline{S}\in Sub_{cl}(\underline{\Sigma})$ is defined as
	\begin{align*}
		\neg \underline{S} \coloneq (\underline{S}\rightarrow \underline{0}).
	\end{align*}
	The contextwise expression for $\neg\underline{S}$ is given by
	\begin{align*}
		(\neg\underline{S})_V = \{\lambda\in\underline{\Sigma}| \forall V'\subseteq V\colon \lambda|_{V'}\notin \underline{S}_{V'} \}.
	\end{align*}
	
	To better align with the adjunction formalism, we have slightly modified D{\"o}ring's definition of the operations on co-Heyting algebra. For any $\underline{S}\in Sub_{cl}(\underline{\Sigma})$, the functor
	\begin{align*}
		\_\vee\underline{S}\colon Sub_{cl}(\underline{\Sigma}) \rightarrow Sub_{cl}(\underline{\Sigma})
	\end{align*}
	has a left adjoint
	\begin{align*}
		\_ \leftarrow \underline{S} \colon Sub_{cl}(\underline{\Sigma}) \rightarrow Sub_{cl}(\underline{\Sigma})
	\end{align*}
	which is referred to as co-Heyting implication. It is characterized by the adjunction
	\begin{align*}
		\underline{T}\leftarrow \underline{S}\leq \underline{R} \quad \text{iff} \quad \underline{T} \leq (\underline{R}\vee \underline{S}),
	\end{align*}
	hence
	\begin{align*}
		(\underline{T}\leftarrow\underline{S})=\bigwedge\{\underline{R}\in Sub_{cl}(\underline{\Sigma})| \underline{T}\leq\underline{R}\vee\underline{S}\}.
	\end{align*}
	We also define a co-Heyting negation $\sim$ for each $\underline{S}\in Sub_{cl}(\underline{\Sigma})$ by
	\begin{align*}
		\sim \underline{S} \coloneq (\underline{\Sigma}\leftarrow \underline{S}).
	\end{align*}
	
	\begin{theorem}
		The definition of the Heyting implication connective is considered trivial for truth-value reasoning, whereas the definition of the co-Heyting implication connective presents non-trivial aspects.
	\end{theorem}
	Here, triviality refers to the failure of the implication connective to yield effective information that would facilitate logical reasoning.
	\begin{proof}
		
		We first examine the definition of Heyting implication connective $\underline{S}\rightarrow\underline{T}$. The proof is conducted under the following cases. 
		
		Case 1: The value of $\underline{S}$ is the empty set in a particular context $V$, denoted as $\underline{S}_V =\emptyset$. In this case, $\underline{S}_V \rightarrow\underline{T}_V =\underline{0}_V \rightarrow\underline{T}_V =\underline{1}_V$. Specifically, if $\forall V, \underline{S}_V=\emptyset$ or equivalently $\underline{S}=\underline{0}$, then $\underline{S}\rightarrow\underline{T}=\underline{0}\rightarrow\underline{T}=\underline{1}$.
		
		Case 2: The value of $\underline{T}$ is the empty set in a particular context $V$ but that of $\underline{S}$ is not, denoted as $\underline{T}_V =\emptyset$. In this case, $\underline{S}_V \rightarrow\underline{T}_V =\underline{S}_V \rightarrow\underline{0}_V =\underline{0}_V=\underline{T}_V$. Specifically, the value of $\underline{T}$ is the empty set in all contexts, while that of $\underline{S}$ is not always the empty set. That is, $\forall V, \underline{T}_V =\emptyset$ or equivalently $\underline{T}=\underline{0}$. Under this setting, $\underline{S}\rightarrow\underline{T}=\underline{S}\rightarrow\underline{0}=\underline{0}=\underline{T}$.
		
		Case 3: The value of $\underline{S}$ in all contexts is a subset of $\underline{T}$, expressed as $\forall V, \underline{S}_V \subseteq\underline{T}_V$. Consequently, $\underline{S}\leq\underline{T}$, and it is evident that $\underline{S}\rightarrow\underline{T}=\underline{1}$ holds.
		
		
		Case 4: Consider the values of $\underline{S}$ and $\underline{T}$ in the minimal context $V_{min}$, where the minimal context has the form of $V_{min}=\{0,\hat{P},\hat{P}^{\bot},1\}$. Excluding the four cases where $\underline{S}$ and $\underline{T}$ respectively take $\underline{0}$(or $\underline{1}$) and $\underline{1}$(or $\underline{0}$), which are included in the cases mentioned above, assume $\underline{S}_{V_{min}}=\{\lambda_{\hat{P}}\}$ and $\underline{T}_{V_{min}}=\{\lambda_{\hat{P}^{\bot}}\}$. In this case, it is evident that $(\underline{S}\rightarrow\underline{T})_{V_{min}}=\{\lambda_{\hat{P}^{\bot}}\}=\underline{T}_{V_{min}}$.
		
		Case 5: The general case. Consider the values of $\underline{S}$ and $\underline{T}$ in a non-minimal context $V$, where $\underline{S}_V$ is not a subset of $\underline{T}_V$. Without loss of generality, let $\underline{S}_V = \{\lambda_i\}_{i\in I}$ and $\underline{T}_V = \{\lambda_j\}_{j\in J}$ (where $I,J$ for index set). For all $\lambda\in \underline{T}_V$, $\lambda$ satisfies the constraint of $(\underline{S}\rightarrow\underline{T})_V$ obviously. For all $\lambda\notin\underline{T}_V$, if $\lambda\in\underline{S}_V$, $\lambda$ does not satisfy the constraint of $(\underline{S}\rightarrow\underline{T})_V$ in $V$. If $\lambda\notin\underline{S}_V$, one can always find a $V'\subseteq V$ and some $\lambda_i \in \underline{S}_V$ satisfying $\lambda |_{V'} = \lambda_i |_{V'}$ such that  $\lambda$ does not satisfy the constraint of $(\underline{S}\rightarrow\underline{T})_V$ in $V'$. Therefore $(\underline{S}\rightarrow\underline{T})_V = \underline{T}_V$ in $V$.
		
		By the attributes of presheaf structure, it can be understood that if $(\underline{S}\rightarrow\underline{T})_V =\underline{1}_V$ holds within context $V$, then $(\underline{S}\rightarrow\underline{T})_{V'} =\underline{1}_{V'}$ holds for all $V'\subseteq V$.

		In summary, aside from the contexts where $\underline{S}_V \subseteq\underline{T}_V$ (including $\underline{S}_V=\emptyset$), which satisfy $\underline{S}_V \rightarrow\underline{T}_V =\underline{1}_V$, in all other instances, $\underline{S}_V \rightarrow\underline{T}_V =\underline{T}_V$. In other words, within any given context $V$, the outcome of $(\underline{S} \rightarrow\underline{T})_V$ can only be either $\underline{T}_V$ or $\underline{1}_V$. Moreover, once $(\underline{S} \rightarrow\underline{T})_V=\underline{1}_V$ holds in a certain context $V_1$, then for all $V_{1}^{'}\subseteq V_1$, it holds that $\underline{S}_{V_{1}^{'}} \rightarrow\underline{T}_{V_{1}^{'}}=\underline{1}_{V_{1}^{'}}$.

		Meanwhile, consider the definition of co-Heyting implication $(\underline{T}\leftarrow\underline{S})$. The general case will suffice. Generally, suppose $\underline{S}_V = \{\lambda_i\}_{i\in I}$ and $\underline{T}_V = \{\lambda_j\}_{j\in J}$. Take $\{\lambda_k\}_{k\in K} = \{\lambda_j\}_{j\in J}\backslash\{\lambda_i\}_{i\in I}$, and we have $\{\lambda_j\}|_{V'}\subseteq \{\lambda_k\}|_{V'}\cup \{\lambda_i\}|_{V'}$ for every $V'\subseteq V$. Hence $(\underline{T}\leftarrow\underline{S})=\bigwedge\{\underline{R}\in Sub_{cl}(\underline{\Sigma})| \forall V, \  \underline{T}_V\backslash\underline{S}_V \subseteq \underline{R}_V \}$.
	\end{proof}
	
	Specifically, we can examine the definition of Heyting negation $\neg \underline{S} \coloneq (\underline{S}\rightarrow \underline{0})$. Given that $\underline{S}_V = \{\lambda_i\}_{i\in I}$ for some $V$, we consider $\neg\underline{S}_V$. For every $\lambda_j \in \underline{\Sigma}_V$, if $\lambda_j \in \{\lambda_i\} _{i\in I}$, then $\lambda_j$ does not satisfy the constraint of $\neg\underline{S}_V$ in $V$. If $\lambda_j \notin \{\lambda_i\} _{i\in I}$, then one can always choose a $V'\subseteq V$ satisfying $\lambda_i |_{V'} = \lambda_j |_{V'}$ such that $\lambda_j$ does not satisfy the constraint of $\neg\underline{S}_V$ in $V'$. Hence $\neg\underline{S}_V$ could only be the empty set in $V$.
	
	In the above discussion, two premises are involved: firstly, $\underline{S}$ must select at least one $\lambda$ within the context $V$, and secondly, $V$ must possess a non-trivial subcontext. Addressing these two premises, we examine two special cases. For the first, consider $\underline{S}=\underline{0}$, where $\underline{0}$ does not select any $\lambda$ within $V$, thereby precluding further discussion. It is evident that $\neg\underline{0}=\underline{1}$. For the second case, it is well-known that when considering all subalgebras (i.e., contexts) of $\mathcal{V}(\mathcal{H})$, the context $V=\{0,1\}$ is excluded. Thus, the minimal context in our discussion is in the form of $V_{min}=\{0,\hat{P},\hat{P}^{\bot},1\}$. Consider $\underline{S}$ selecting $\lambda$ in $V_{min}$. Suppose $\underline{S}_{V_{min}}=\{\lambda_{\hat{P}}\}$, then $\neg \underline{S}_{V_{min}}=\{\lambda_{\hat{P}^{\bot}}\}$, and no smaller context $V'$ than $V_{min}$ could lead to a contradiction, which indicates $\neg\underline{S}\neq \underline{0}$. Therefore, aside from these two special cases, for all other situations, the discussion holds true, namely, $\neg\underline{S}=\underline{0}$.
	
	Similarly, consider the definition of co-Heyting negation
	\begin{align*}
		\sim\underline{S} = \bigwedge\{\underline{R}\in Sub_{cl}(\underline{\Sigma})| \underline{\Sigma}\leq\underline{R}\vee\underline{S}\}=\bigwedge\{\underline{R}\in Sub_{cl}(\underline{\Sigma})| \underline{\Sigma} = \underline{R}\vee\underline{S}\}.
	\end{align*}
	Provided that $\underline{S}_V = \{\lambda_i\}_{i\in I}$, we have $\sim\underline{S}_V = \underline{\Sigma}_V \backslash \{\lambda_i\}_{i\in I} :=\{\lambda_j\}_{j\in J}$. For any $V'\subseteq V$, the constraint $\underline{\Sigma}_{V'} = \sim\underline{S}_{V'}\vee\underline{S}_{V'}$ will always be satisfied.
	
	We note that $\neg\underline{S}$, as well as Heyting implication, is defined by satisfying the law of non-contradiction $\underline{S}\wedge\neg\underline{S} = 0$, while $\sim\underline{S}$, as well as co-Heyting implication, is defined by satisfying the law of excluded middle $\underline{S}\vee\sim\underline{S} = 1=\underline{\Sigma}$.
	
	In this paper, we aim to construct the truth-value reasoning that is suitable for Heyting algebra. To achieve this, we intend to select elements in contexts based on the law of non-contradiction. By doing so, we can choose elements that satisfy the logical connectives in Heyting algebra.

	\subsection{Heyting algebra structure}
	In our perspective, for any projection $\hat{P}$, the role of daseinisation is to find the ``approximation" of $\hat{P}$ in every context. For a given context $V$, if $\hat{P}\in V$, the approximation of $\hat{P}$ in the context is $\hat{P}$ itself. If $\hat{P}\notin V$, then the approximation of $\hat{P}$ is the element closest to $\hat{P}$ in the context. The outer daseinisation of $\hat{P}$ is the element just larger than $\hat{P}$, while the inner daseinisation is the element just smaller than $\hat{P}$. We regard the daseinisation of $\hat{P}$ as an ``image" of $\hat{P}$ in every context.
	\begin{proposition}\label{prop1}
		For any $\hat{P}\in L$ and any $V\in \mathcal{V}(\mathcal{H})$, we have $\delta^i(\hat{P})_V \bot \delta^o(\hat{P}^{\bot})_V$, i.e., $\delta^i(\hat{P})\bot\delta^o(\hat{P}^{\bot})$ holds for every context (namely contextwisely).
	\end{proposition}
	
	\begin{proof}
		For any $V\in \mathcal{V}(\mathcal{H})$. Let us define $\tilde{P}:=\delta^i (\hat{P})_V$. We have $\tilde{P}\in V$ and $\tilde{P}\leq \hat{P}$. Furthermore, for any $\hat{P}_V\in V$, if $\hat{P}_V \leq\hat{P}$, then $\hat{P}_V \leq\tilde{P}$.
		
		Now, let's define $\tilde{Q}:=\delta^o (\hat{P}^{\bot})_V$. We have $\tilde{Q}\in V$ and $\tilde{Q}\geq \hat{P}^{\bot}$. Similarly, for any $\hat{Q}_V \in V$, if $\hat{Q}_V \geq \hat{P}^{\bot}$, then $\hat{Q}_V \geq \tilde{Q}$.
		
		Since $\tilde{P}, \tilde{Q}\in V$, we have $\tilde{P}\bot\tilde{Q}$ if and only if $\tilde{P}\wedge\tilde{Q}=0, \tilde{P}\vee\tilde{Q}=1$.
		
		Assume that $\tilde{P}\wedge\tilde{Q}=R > 0$. Then we have $R^{\bot}=\tilde{P}^{\bot}\vee\tilde{Q}^{\bot}$. Since $R\leq \tilde{P}$ and $\tilde{P}\leq\hat{P}$.	We can imply that $R\leq \hat{P}$, and $R^{\bot}\geq \hat{P}^{\bot}$. Hence $R^{\bot}\geq \tilde{Q}$. 
		
		Since $R^{\bot}=\tilde{P}^{\bot}\vee\tilde{Q}^{\bot}$, we have $R^{\bot}\wedge\tilde{Q}=(\tilde{P}^{\bot}\vee\tilde{Q}^{\bot})\wedge\tilde{Q}=\tilde{P}^{\bot}\wedge\tilde{Q}$. Since $R^{\bot}\geq \hat{P}^{\bot}, \tilde{Q}\geq\hat{P}^{\bot}$. We can conclude that $R^{\bot}\wedge\tilde{Q}\geq\hat{P}^{\bot}$ and $R^{\bot}\wedge\tilde{Q}\geq\tilde{Q}$. Consequently, $\tilde{P}^{\bot}\wedge\tilde{Q}\geq\tilde{Q}$. Since $\tilde{P}^{\bot}\wedge\tilde{Q}\leq\tilde{Q}$, we have $\tilde{P}^{\bot}\wedge\tilde{Q} = \tilde{Q}$. Therefore $\tilde{Q}\leq\tilde{P}^{\bot}$. As a result, we can imply that $\tilde{P}\wedge\tilde{Q}\leq\tilde{P}\wedge\tilde{P}^{\bot}=0$. This contradicts the assumption that $\tilde{P}\wedge\tilde{Q}=R > 0$. Therefore $\tilde{P}\wedge\tilde{Q}=0$ holds. 
		
		Similarly, $\tilde{P}\vee\tilde{Q}=1$ holds. We can imply that $\tilde{P}\bot\tilde{Q}$.
	\end{proof}
	\begin{proposition}\label{prop2}
		$\delta^i(\hat{P})_V\leq\delta^o(\hat{P})_V$, i.e. $\delta^i(\hat{P})\leq\delta^o(\hat{P})$ holds for every context.
	\end{proposition}
	
	\begin{proof}
		This proposition can be easily verified by the definition of daseinisation.
	\end{proof}
	
	It should be noted that we can utilize the conclusion of Proposition \ref{prop2} to prove Proposition \ref{prop1}. From Proposition \ref{prop2}, we have $\delta^i(\hat{P})\leq \hat{P}\leq \delta^o(\hat{P})$. By orthogonal complement, it follows that $\hat{P}^{\bot}\leq \delta^i(\hat{P})^{\bot}$ holds for every context. Furthermore, since $\delta^o(\hat{P}^{\bot})$ represents the least upper bound of $\hat{P}^{\bot}$ across every context, we obtain $\delta^o(\hat{P}^{\bot})\leq \delta^i(\hat{P})^{\bot}$ for every context. Consequently, it satisfies $\delta^i(\hat{P})\wedge \delta^o(\hat{P}^{\bot})\leq \delta^i(\hat{P})\wedge \delta^i(\hat{P})^{\bot}=0$. Therefore, $\delta^i(\hat{P})$ is orthogonal to $\delta^o(\hat{P}^{\bot})$, which implies that Proposition \ref{prop1} holds.
	
	Consider two properties $\hat{P}$ and $\hat{P}^{\bot}$ being orthogonal complements. By the law of non-contradiction, when we assign $\hat{P}$ to true, $\hat{P}^{\bot}$ must be assigned false. When assigning truth values in every context $V$, the truth assignment must make the ``reflection" of $\hat{P}$ true, and that of $\hat{P}^{\bot}$ false. According to the above two propositions, assigning true to $\delta^i(\hat{P})_V$ in context $V$ meets the requirement. Thus, we propose assigning $\delta^i(\hat{P})_V$ to true in every context $V$ if we would like to assign $\hat{P}$ to true. 
	
	\begin{proposition}
		In the sense of sets, we have $\delta^o (\hat{P})=\uparrow \hat{P}$, and $\delta^i (\hat{P})=\downarrow \hat{P}$.
	\end{proposition}
	\begin{proof}
		$\forall V\in \mathcal{V}(\mathcal{H})$, we have $\hat{P}\leq\delta^o (\hat{P})_V$, hence $\delta^o (\hat{P})_V\in\uparrow\hat{P}$. It follows that $\delta^o (\hat{P})\subseteq\uparrow\hat{P}$.
		\\
		$\forall \hat(Q)\in\uparrow\hat{P}$, we choose $\{\hat{Q},\hat{Q}^{\bot}\}$ as the generators to construct the context $V_{\hat{Q}}$ such that $\delta^o (\hat{P})_{V_{\hat{Q}}}=\hat{Q}$. It follows that $\uparrow\hat{P}\subseteq\delta^o (\hat{P})$.	Therefore $\delta^o (\hat{P})=\uparrow \hat{P}$.
		\\
		Similarly, we have $\delta^i (\hat{P})=\downarrow \hat{P}$.
	\end{proof}
	
	Given a property lattice $L$, we start with constructing a lattice $I(L)$ by taking the order ideal of every element in $L$. It is well-known that $L$ is isomorphic to $I(L)$. Next, we add disjunctive elements to $I(L)$ through MacNeille completion, obtaining the lattice $D(L)$, where the elements are downsets of one or more elements in $L$. We have an embedding $I(L)\hookrightarrow D(L)$, and the added disjunctive elements are neither redundant nor missing. Any element in $D(L)$ that is not in the form of an ideal is a union of ideals, and all unions of different ideals are in $D(L)$.
	
	Inspired by Dan Marsden\cite{marsden2010comparing}, we can consider ``globality" and ``contextuality" separately. For outer daseinisation, we denoted by $\uparrow\hat{P}$ the set of all true properties that $\hat{P}$ being true implies. These properties are bound to be true whenever they appear in any context. $\uparrow\hat{P}\cap P(V)$ represents all true properties in context $V$ (when $\hat{P}$ is assigned true), where $P(V)$ denotes all projections in $V$. $\bigwedge\{\uparrow\hat{P}\cap P(V)\}=\delta^o (\hat{P})_V$ holds.
	
	For inner daseinisation, we denoted by $\downarrow\hat{P}$ the set of all properties that imply $\hat{P}$ to be true. These properties are true in specific contexts, but not in all contexts. $\downarrow\hat{P}\cap P(V)$ represents a set of properties that imply $\hat{P}$ to be true in context $V$. $\bigvee\{\downarrow\hat{P}\cap P(V)\}=\delta^i (\hat{P})_V$ holds. That is, in context $V$, we assign the set $\downarrow\hat{P}\cap P(V)$ to true. Note that we only know that the top element $\delta^i (\hat{P})_V$ of the set $\downarrow\hat{P}\cap P(V)$ is true, but we do not know the truth value of other elements. 
	
	\begin{definition}
		Let $L$ be a property lattice, and define the mapping
		\begin{align*}
			\overline{\delta^i}:L&\to D(L) \\
			\hat{P}&\mapsto \delta^i (\hat{P})=\downarrow\hat{P}
		\end{align*}
		where $\delta^i (\hat{P})=\downarrow\hat{P}$ is understood in terms of sets.
	\end{definition}
	
	Similar to $\underline\delta^o$, $\overline{\delta^i}$ preserves all meets and is an order-preserving injection, but not a surjection. It is evident that $\overline{\delta^i(\hat{0})}=\downarrow\hat{0}=\{\hat{0}\}$, and $\overline{\delta^i (\hat{1})}=\downarrow\hat{1}=L$. For meets, we have
	\begin{align*}
		\forall \hat{P},\hat{Q}\in L:\overline{\delta^i (\hat{P})}\wedge\overline{\delta^i (\hat{Q})}=\overline{\delta^i (\hat{P}\wedge\hat{Q})}.
	\end{align*}
	However, for joins, we have
	\begin{align*}
		\forall \hat{P},\hat{Q}\in L:\overline{\delta^i (\hat{P})}\vee\overline{\delta^i (\hat{Q})}\leq\overline{\delta^i (\hat{P}\vee\hat{Q})}.
	\end{align*}
	
	In general, $\overline{\delta^i (\hat{P})}\vee\overline{\delta^i (\hat{Q})}$ does not have the form of $\overline{\delta^i (\hat{R})}$ for any projection $\hat{R}\in L$.
	
	Next, we show that the elements in $D(L)$ constitute a complete Heyting algebra. For a downset $\overline{S}\in D(L)$, we define $\overline{S}_V :=\alpha_V( \bigvee\{\overline{S}\cap P(V)\})$. Let $\overline{S},\overline{T}\in D(L)$ be two downsets, then
	\begin{align*}
		\forall V\in \mathcal{V}(\mathcal{H}):(\overline{S}\wedge\overline{T})_V &= \overline{S}_V \cap \overline{T}_V\\
		(\overline{S}\vee\overline{T})_V &= \overline{S}_V \cup \overline{T}_V.
	\end{align*}
	
	For each $\overline{S}\in D(L)$, the functor
	\begin{align*}
		\overline{S}\wedge \_ :D(L)\to D(L)
	\end{align*}
	has a right adjoint
	\begin{align*}
		\overline{S}\to \_ : D(L)\to D(L).
	\end{align*}
	The Heyting implication is given by the adjunction
	\begin{align*}
		\overline{S}\wedge\overline{R}\leq\overline{T} \quad \text{iff} \quad \overline{R}\leq(\overline{S}\to\overline{T}).
	\end{align*}
	This implies that
	\begin{align*}
		(\overline{S}\to\overline{T})&=\bigvee\{\overline{R}\in D(L) | \overline{S}\wedge\overline{R}\leq\overline{T}\}\\
		&=\{r\in L | \forall a\in\overline{S}, a\wedge r\in\overline{T}\}.
	\end{align*}
	The contextwise definition is: for every $V\in \mathcal{V}(\mathcal{H})$,
	\begin{align*}
		(\overline{S}\to\overline{T})_V =\bigvee\{(\overline{S}\to\overline{T})\cap P(V)\}.
	\end{align*}
	The Heyting negation $\neg$ is defined for every $\overline{S}\in D(L)$ as follows:
	\begin{align*}
		\neg\overline{S}:=(\overline{S}\to\overline{0}).
	\end{align*}
	
	Therefore, we propose the following theorem.
	\begin{theorem}
		The elements of $D(L)$ constitute a complete Heyting algebra.
	\end{theorem}
	
	\begin{proof}
		For any $\overline{R},\overline{S},\overline{T}\in D(L)$ and any $V\in \mathcal{V}(\mathcal{H})$, we have
		\begin{align*}
			(\overline{R}\wedge(\overline{S}\wedge\overline{T}))_{V}=\overline{R}_{V}\cap (\overline{S}_{V}\cap\overline{T}_{V}) = (\overline{R}_{V}\cap\overline{S}_{V})\cap\overline{T}_{V}=((\overline{R}\wedge\overline{S})\wedge\overline{T})_{V}
		\end{align*}
		and
		\begin{align*}
			(\overline{R}\vee(\overline{S}\vee\overline{T}))_{V}=\overline{R}_{V}\cup (\overline{S}_{V}\cup\overline{T}_{V}) = (\overline{R}_{V}\cup\overline{S}_{V})\cup\overline{T}_{V}=((\overline{R}\vee\overline{S})\vee\overline{T})_{V}.
		\end{align*}
		Thus, the commutative and associative laws hold in $D(L)$.
		
		In $D(L)$, there exist the maximum element $1:=\downarrow\hat{1}$ and the minimum element $0:=\downarrow\hat{0}$, where $\downarrow\hat{1}$ represents the set of all properties in the property lattice $L$ and $\downarrow\hat{0}$ represents the set of the bottom element $\hat{0}$ of $L$.
		
		The Heyting implication is given by the above adjunction:
		\begin{align*}
			\overline{S}\wedge\overline{R}\leq\overline{T} \quad iff \quad \overline{R}\leq(\overline{S}\to\overline{T}).
		\end{align*}
		
		Thus, the elements of $D(L)$ constitute a complete Heyting algebra.
	\end{proof}
	It is worth noting that this Heyting algebra coincides with the one defined by Bob Coecke in \cite{coecke2002quantum}. The Heyting algebra constructed by Coecke indeed corresponds to inner daseinisation. A downset $\downarrow\hat{P}$ is associated with an actuality set\textsuperscript{\cite{coecke2002quantum}}, and the Heyting algebra generated by all downsets is the injective hull of the property lattice. 
	
	In outer daseinisation, we assign true to every element in $\uparrow\hat{P}$. Specifically, in each context $V$, every element in the set $\uparrow\hat{P}\cap P(V)$ is assigned as true. The conjunction $\bigwedge\{\uparrow\hat{P}\cap P(V)\}$ is definitely assigned as true. In inner daseinisation, we assign the set $\downarrow\hat{P}$ as true. In each context $V$, the set $\downarrow\hat{P}\cap P(V)$ is assigned as true. However, only the element $\bigvee\{\downarrow\hat{P}\cap P(V)\}$ is determined to be true in the current context. 
	
	Each property in the property lattice, as a generator of a certain context, corresponds to a linear functional in a one-to-one manner. When a property is not one-dimensional, its splitting is usually not unique. The splitting of a property represents the join of generators of properties in a particular context, declaring the context in which the property exists. The Heyting algebra $D(L)$ is generated by downsets of elements of the property lattice, where the additional disjunctive elements represent the declaration that a property is true only in those specific contexts, rather than in all contexts.

	\section{Natural revision in topos quantum logic}
	We propose a reasoning system that combines the characteristics of two reasoning modes, static and dynamic, as follows: static intuitionistic quantum logic is used as the reasoning rules, and dynamic operations induced by actions are used for truth value assignments.
	
	
	As Coecke et al.\cite{coecke2002quantum,coecke2004logic,coecke2004sasaki} have demonstrated, a projection measurement is an external operation, denoted as $\varphi_a^*$, parameterized by the projector $a$. This operation acts upon elements within the property lattice, leading to new properties. Moreover, this operation can be extended to the Heyting algebra generated by property sets. When the projection measurement operation is applied to elements in the Heyting algebra that are assigned true (corresponding to the downset of some properties), it can alter the truth value assignments of elements within the Heyting algebra, resulting in new elements being assigned true. For the sake of simplification and provided that it does not lead to any ambiguity, we do not differentiate between the symbols of the $\varphi_a^*$ operation when applied to either the property lattice or the Heyting algebra.
	
	We use the contextuality to map properties to the corresponding elements in the Heyting algebra and perform the appropriate truth value reasoning. Regarding the changes in truth value assignments caused by external actions, elements within the Heyting algebra that were assigned true in the preceding moment are mapped to elements assigned true in the subsequent moment according to the operational rules of projection measurements.
	
	A consequence relation $\vdash\subseteq Fml\times Fml$ is determined by a binary relation obtained from antecedents and consequents.
	In Heyting algebra $D(L)$, we take the elements that are assigned to true as the antecedents of a consequence relation, and the elements that are inferred to be true according to the partial order relation based on the true-valued elements as the consequents. A truth-value assignment in the Heyting algebra corresponds to the antecedents that are associated with a consequence relation. By adding new messages, the antecedents of a consequence relation are changed, which is equivalent to revising one consequence relation to another.
	
	When examining natural revision based on quantum logic, we should consider operations in the two types of reasoning - static and dynamic - as two revision operators for consequence relations. We suppose that a consequence relation $\vdash\subseteq Fml\times Fml$ occurs within Heyting algebra $D(L)$, where the formula set $Fml$ represents elements in the Heyting algebra, i.e., $Fml=D(L)$. There are two revision operators for revising the consequence relation: one is to add new message using the truth-value reasoning of static quantum logic, and the other is to add new message by dynamic changes of truth-value assignments.
	
	\begin{lemma}
		The $\wedge$ operation within a Heyting algebra can be seen as a type of revision operation, referred to as static revision. Every formula $\alpha\in D(L)$ induces an operator on $\mathcal{C}$
		\begin{align*}
			\bar{\alpha}:\mathcal{C}\to\mathcal{C}
		\end{align*}
		via
		\begin{align*}
			\bar{\alpha}(\vdash_1)=\vdash_2 \quad \Leftrightarrow \quad \alpha\wedge A_1 = A_2
		\end{align*}
		in which $A_1$ and $A_2$ represents antecedents of $\vdash_1$ and $\vdash_2$ respectively.
	\end{lemma}
	\begin{proof}
		For any $\overline{S}\in D(L)$, we have $\overline{S}\subseteq L$ and $1=L$. Hence, $1\wedge\overline{S}=\overline{S}\wedge 1=\overline{S}$ holds.
		
		By the definition of Heyting implication, $\overline{S}\wedge \_ \dashv \overline{S} \to \_$ holds.
		
		For any $\overline{S},\overline{T}\in D(L)$ and any $V\in\mathcal{V(\mathcal{H})}$, $(\overline{S}\wedge\overline{T})_{V}=\overline{S}_{V}\cap\overline{T}_{V}\subseteq\overline{S}_{V}$. So $\overline{S}\wedge\overline{T}\leq\overline{S}$ holds.
		
		If $\overline{S}\leq\neg\overline{T}$, then for any $x\in\overline{S}$, we have $x\in\neg\overline{T}$. $\neg\overline{T}=(\overline{T}\to\overline{0})=\{r\in L|\forall a\in\overline{T},a\wedge r=0\}$. So, for any $x\in\overline{S}$ and any $a\in \overline{T}$, we have $a\wedge x=0$. Hence, $\overline{S}\wedge\overline{T}=0$ holds. Conversely, if $\overline{S}\wedge\overline{T}=0$, then for any $x\in\overline{S}$ and any $a\in \overline{T}$, we have $x\wedge a=0$. So, $x\in \neg\overline{T}$ holds for any $x\in \overline{S}$. Hence, $\overline{S}\leq\neg\overline{T}$ holds.
		
		$\overline{T}\leq\overline{S}\Leftrightarrow \forall V\in\mathcal{V(\mathcal{H})},\overline{T}_{V}\subseteq\overline{S}_{V}\Leftrightarrow \forall V\in\mathcal{V(\mathcal{H})},\overline{S}_{V}\cap\overline{T}_{V}=\overline{T}_{V}\Leftrightarrow \overline{S}\wedge\overline{T}=\overline{T}$.
	\end{proof}
	This revision operation takes the elements within Heyting algebra $D(L)$ as operators and performs the $\wedge$ operation of Heyting algebra on new message and the antecedents of the original consequence relation to obtain new antecedents within the same lattice. These new antecedents correspond to a new consequence relation.
	
	\begin{lemma}
		The $\wedge_s$ operation outside a Heyting algebra can be viewed as another type of revision operation, referred to as dynamic revision. Every formula corresponding to a closed space $a\in L$ induces an operator on $\mathcal{C}$
		\begin{align*}
			\bar{\varphi^*_a}:\mathcal{C}\to\mathcal{C}
		\end{align*}
		via
		\begin{align*}
			\bar{\varphi^*_a}(\vdash_{D(L)_1})=\vdash_{D(L)_2} \quad \Leftrightarrow \quad a\wedge_s A_{D(L)_1} = A_{D(L)_2}
		\end{align*}
		in which $A_{D(L)_1}$ and $A_{D(L)_2}$ represents antecedents of $\vdash_{D(L)_1}$ and $\vdash_{D(L)_2}$ respectively.
	\end{lemma}
	\begin{proof}
		For any $a\in L$, we have $1\wedge_s a=1\wedge(\neg 1\vee a)=1\wedge(0\vee a)=1\wedge a=a$ and $a\wedge_s 1=a\wedge(\neg a\vee 1)=a\wedge 1=a$.
		
		By the definition of Sasaki adjoint, $a\wedge_s \_\dashv a\rightsquigarrow_s \_$ holds.
		
		For any $a,b\in L$, $a\wedge_s b=a\wedge(\neg a\vee b)\leq a$ holds.
		
		If $a\leq\neg b$, then $a\leq\neg b\Rightarrow b\leq \neg a\Rightarrow \neg a\vee b=\neg a\Rightarrow a\wedge_s b=a\wedge (\neg a\vee b)=a\wedge\neg a=0$. Conversely, if $a\wedge_s b=a\wedge(\neg a\vee b)=0$, then $\neg a=\neg a\vee 0=\neg a\vee (a\wedge (\neg a\vee b))$. Clearly $\neg a\leq \neg a\vee b$. By the orthomodular law, we have $\neg a\vee (a\wedge (\neg a\vee b))=\neg a\vee b$. Thus, $\neg a=\neg a\vee b\Rightarrow b\leq \neg a\Rightarrow a\leq \neg b$.
		
		If $b\leq a$, then by the dual orthomodular law, $a\wedge(\neg a\vee b)=b$. Thus, $a\wedge_s b=a\wedge(\neg a\vee b)=b$. Conversely, if $a\wedge_s b=a\wedge(\neg a\vee b)=b$, then $a\wedge(\neg a\vee b)\leq a$ holds. Thus $b\leq a$ holds. 
	\end{proof}
	This revision operation takes a element $a$ in a property lattice as operator and performs the Sasaki projection $\varphi^*_a$ on new message and the antecedents of the original consequence relation in a lattice to obtain new antecedents within another lattice. These new antecedents also correspond to a new consequence relation.
	
	It is worth noting that elements in the two lattices are identical, but the truth-value assignments of the elements may differ.
	
	\begin{theorem}
		Given a consequence relation $\vdash$ within a Heyting algebra $H$. Let $\alpha\in H$ and $a\in L$ be formulas. There are two types of revision operators, called static and dynamic revision respectively, to revise the consequence relation. The former, denoted by $\bar{\alpha}$, maps the consequence relation to another one within the same lattice $H$. The latter, denoted by $\overline{\varphi^*_a}$, maps the consequence relation to another one within a different lattice.
	\end{theorem}
	\begin{proof}
		It is straightforward to demonstrate by using the above two lemmas.
	\end{proof}

	When we use formulas as operators to revise the consequence relation, there are two sources of formulas: those from within the Heyting algebra and those from outside. When a formula from within the Heyting algebra is used as an operator to revise the consequence relation, the revision operator is applied to the antecedents of the consequence relation based on the $\wedge$ operation within Heyting algebra. This results in a new consequence relation within the same lattice. When a formula from outside the Heyting algebra is used as an operator, the selected formula corresponds to a closed subspace, namely a projection operation, and is used as a parameter $a$ with Sasaki projection $\varphi_a^*$ operation to revise the antecedents of the original consequence relation. This results in a new consequence relation within another lattice.
	
	For example, consider the two equations below.
	\begin{align*}
		\varphi^*_a(B\wedge C) = D \qquad B\wedge\varphi^*_a(C) = E
	\end{align*}
	In a Heyting algebra, there exists an element $C$ that is assigned to true, which corresponds to a consequence relation $\vdash_C$ based on an antecedent $\{C\}$. For the former equation, upon adding new message ``element $B$ is assigned to true", a static revision operation $\wedge$ is performed on consequence relation $\vdash_C$ to obtain antecedents $\{B, C\}$, which correspond to consequence relation $\vdash_{B\wedge C}$. Then, upon adding new message ``projecting onto space $a$", a dynamic revision operation $\wedge_s$ is employed to revise consequence relation $\vdash_{B\wedge C}$, resulting in an antecedent $\{\varphi^*_a(B\wedge C) = D\}$ and corresponding to consequence relation $\vdash_D$. For the latter, firstly upon adding new message ``projecting onto space $a$", a dynamic revision operation $\wedge_s$ is employed to revise consequence relation $\vdash_C$, resulting in an antecedent $\{\varphi^*_a(C)\}$ and corresponding to consequence relation $\vdash_{\varphi^*_a(C)}$. Then, upon adding new message ``element $B$ is assigned to true", a static revision operation $\wedge$ is performed on consequence relation $\vdash_{\varphi^*_a(C)}$ to obtain an antecedent $\{B\wedge\varphi^*_a(C) = E\}$, which corresponds to consequence relation $\vdash_E$. As shown in the Figure 1.
	
	\begin{figure}[H]
		\centering
		\includegraphics[width=\textwidth]{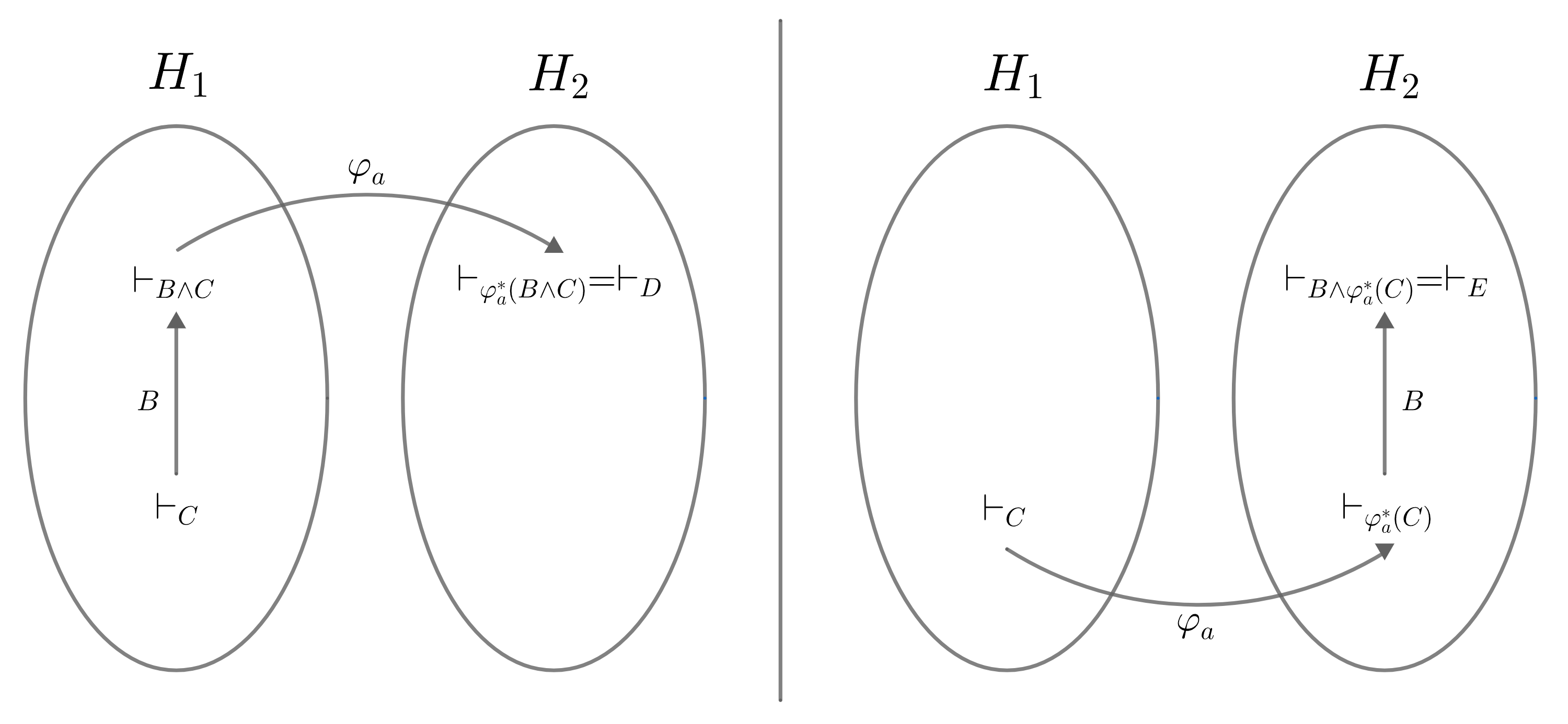}
		\caption{Natural revision}
		\label{Fig.Revision}
	\end{figure}
	
	Both static and dynamic revision involve changing the truth-value assignments of elements in the Heyting algebra. The difference lies in that static revision directly adds a new element assigned to be true, and performs $\wedge$ operation on the original elements to change their truth values. On the other side, although it is known from $a\wedge_s b\leq a$ that the new element $a$ is assigned to true, for dynamic revision, the new element is applied to elements of the original lattice by $\wedge_s$ operation to obtain elements of another lattice which is assigned to true.
	
	If we regard natural revision as a consequence revision system (CRS), this system should have two revision operations, $\wedge$ and $\wedge_s$, which correspond to two internalizing connectives, $\rightarrow$ and $\rightsquigarrow_s$ respectively. However, for the $a\wedge_s\_$ operation, since the element $a$ is a parameter of external operation, it is denoted as $\varphi_a^*(\_)$. Correspondingly, the $a$ in the internalizing connective $a\rightsquigarrow_s\_$ should also be a parameter, denoted as $\varphi_{a,*}(\_)$. In this manner, the revision operator and so-called internalizing connective obtained by Sasaki adjoint are not suitable for the formalism of internalizing connectives. That is, the modus ponens obtained by Sasaki adjoint in standard quantum logic is not suitable for ordinary truth-value reasoning. The consequence relation only occurs within the Heyting algebra, namely the formula set $Fml$ in $\vdash\subseteq Fml\times Fml$ needs to be selected from the elements in the same lattice. The Heyting algebra has the $\wedge$ operation, and new message can transform the formula set $Fml_1$ into another formula set $Fml_2$ in the same lattice. On the other hand, the Sasaki projection is an external operation of the Heyting algebra, and new message can transform one formula set $Fml_{H_1}$ in the lattice $H_1$ into another formula set $Fml_{H_2}$ in another lattice $H_2$.
	
	
	\section*{Conclusion}
	
	In this paper, we point out that the revision theory based on standard quantum logic cannot correspond to truth-value reasoning in quantum logic. The reasoning rules generated by Sasaki adjoint in standard quantum logic are not truth-value reasoning within the system, but truth-value changes obtained by taking external elements as parameters of the operation. To accurately describe truth-value reasoning in quantum logic, it is necessary to reconsider truth-value assignment in quantum logic. In this paper, we use topos quantum theory and Heyting algebra to characterize the truth-value reasoning of quantum logic and indicate that the methodology proposed by Isham et al. in topos quantum theory does not apply to truth-value reasoning in the Heyting algebra sense. To give truth-value reasoning that satisfies Heyting algebra, we redefine the structure in topos quantum theory by taking the downset of a property as its inner daseinisation in all contexts, construct a Heyting algebra with downsets as elements, and assign truth values to downsets satisfies the law of non-contradiction. In this formalism, truth-value reasoning in the sense of Heyting algebra is carried out. Moreover, we regard the Sasaki adjoint as an operation that selects a property from the property lattice as a parameter and maps elements in one Heyting algebra to those in another. By using internal and external operations in Heyting algebra, the reasoning rules of quantum logic can be characterized. Quantum logic is determined by both static and dynamic reasoning.
	
	After giving truth-value assignment and reasoning rules for quantum logic, we consider giving a natural revision theory based on quantum logic. We correspond the static and dynamic reasoning in quantum logic to two revision operations, which are static and dynamic revision, respectively. The static reasoning corresponds to a revision based on the $\wedge$ operation inside the Heyting algebra. The dynamic reasoning corresponds to a revision based on the $\wedge_s$ operation outside the Heyting algebra. The consequence relation $\vdash\subseteq Fml\times Fml$ arises within a Heyting algebra and the set of formulas $Fml$ is constituted by elements in the lattice. Static revision revises the consequence relation to another consequence relation in the same lattice, while dynamic revision revises the consequence relation within one lattice to another consequence relation in another lattice with the same elements but different truth-value assignments.
	
	The research of truth-value assignment and truth-value reasoning in intuitionistic quantum logic can be extended to mixed states of quantum mechanics. The investigation of how entanglement and superposition phenomena in quantum mechanics manifest in quantum logic is worth studying. In quantum computation, unitary operators represent evolution operations and can also serve as dynamic operators. This further expands the natural revision theory and provides a theoretical foundation for quantum computation.

	\newpage
	\nocite{*}
	\bibliographystyle{spmpsci}
	\bibliography{SubmitRevisionref}

\begin{thebibliography}{10}
\providecommand{\url}[1]{{#1}}
\providecommand{\urlprefix}{URL }
\expandafter\ifx\csname urlstyle\endcsname\relax
  \providecommand{\doi}[1]{DOI~\discretionary{}{}{}#1}\else
  \providecommand{\doi}{DOI~\discretionary{}{}{}\begingroup
  \urlstyle{rm}\Url}\fi

\bibitem{1985On}
Alchourron, C., Gardenfors, P., Makinson, D.: {On the Logic of Theory Change}.
\newblock Journal of Symbolic Logic \textbf{50}, 510--530 (1985)

\bibitem{babu2023quantum}
Babu, H.M.H.: Quantum Computing: A pathway to quantum logic design.
\newblock IOP Publishing (2023)

\bibitem{birkhoff1936logic}
Birkhoff, G., Von~Neumann, J.: {The Logic of Quantum Mechanics}.
\newblock Annals of mathematics pp. 823--843 (1936)

\bibitem{coecke2002quantum}
Coecke, B.: Quantum logic in intuitionistic perspective.
\newblock Studia Logica pp. 411--440 (2002)

\bibitem{coecke2004logic}
Coecke, B., Moore, D.J., Smets, S.: Logic of dynamics and dynamics of logic:
  some paradigm examples.
\newblock Logic, Epistemology, and the Unity of Science pp. 527--555 (2004)

\bibitem{coecke2004sasaki}
Coecke, B., Smets, S.: The sasaki hook is not a [static] implicative connective
  but induces a backward [in time] dynamic one that assigns causes.
\newblock International Journal of Theoretical Physics \textbf{43}(7-8),
  1705--1736 (2004)

\bibitem{dalla2001quantum}
Dalla~Chiara, M.L., Giuntini, R., et~al.: Quantum logics.
\newblock Handbook of philosophical logic \textbf{6}, 129--228 (2001)

\bibitem{doring2011topos}
D{\"o}ring, A.: Topos quantum logic and mixed states.
\newblock Electronic Notes in Theoretical Computer Science \textbf{270}(2),
  59--77 (2011)

\bibitem{doring2016topos}
D{\"o}ring, A., Chubb, J., Eskandarian, A., Harizanov, V.: Topos-based logic
  for quantum systems and bi-heyting algebras.
\newblock Logic and Algebraic Structures in Quantum Computing \textbf{45},
  151--173 (2016)

\bibitem{doring2008topos1}
D{\"o}ring, A., Isham, C.J.: {A topos foundation for theories of physics: I.
  Formal languages for physics}.
\newblock Journal of Mathematical Physics \textbf{49}(5), 053515 (2008)

\bibitem{doring2008topos2}
D{\"o}ring, A., Isham, C.J.: {A topos foundation for theories of physics: II.
  Daseinisation and the liberation of quantum theory}.
\newblock Journal of Mathematical Physics \textbf{49}(5), 053516 (2008)

\bibitem{2002Quantum}
Engesser, K., Gabbay, D.M.: {Quantum logic, Hilbert space, revision theory}.
\newblock Artificial Intelligence \textbf{136}(1), 61--100 (2002)

\bibitem{finch_1969}
Finch, P.: Sasaki projections on orthocomplemented posets.
\newblock Bulletin of the Australian Mathematical Society \textbf{1}(3),
  319–324 (1969).
\newblock \doi{10.1017/S0004972700042192}

\bibitem{1970Quantum}
Finch, P.D.: Quantum logic as an implication algebra.
\newblock Bulletin of the Australian Mathematical Society \textbf{2}(1),
  101--106 (1970)

\bibitem{flori2012lectures}
Flori, C.: Lectures on topos quantum theory.
\newblock arXiv preprint arXiv:1207.1744  (2012)

\bibitem{1983Orthomodular}
Gudrun, K.H.E.: {Orthomodular Lattices}.
\newblock Orthomodular Lattices (1983)

\bibitem{hamilton2000topos}
Hamilton, J., Isham, C.J., Butterfield, J.: {Topos perspective on the
  Kochen-Specker theorem: III. von Neumann algebras as the base category}.
\newblock International Journal of Theoretical Physics \textbf{39}, 1413--1436
  (2000)

\bibitem{isham1998topos}
Isham, C.J., Butterfield, J.: {Topos perspective on the Kochen-Specker theorem:
  I. Quantum states as generalized valuations}.
\newblock International journal of theoretical physics \textbf{37}(11),
  2669--2733 (1998)

\bibitem{kitajima2022negations}
Kitajima, Y.: {Negations and Meets in Topos Quantum Theory}.
\newblock Foundations of Physics \textbf{52}(1), 12 (2022)

\bibitem{kochen1990problem}
Kochen, S., Specker, E.P.: The problem of hidden variables in quantum
  mechanics.
\newblock Ernst Specker Selecta pp. 235--263 (1990)

\bibitem{marsden2010comparing}
Marsden, D.: Comparing intuitionistic quantum logics.
\newblock Ph.D. thesis, University of Oxford. URl: http://www. cs. ox. ac.
  uk/people/bob. coecke/Marsden. pdf (2010)

\bibitem{nielsen2002quantum}
Nielsen, M.A., Chuang, I.: Quantum computation and quantum information (2002)

\bibitem{roman1999quantum}
Roman, L., Zuazua, R.E.: Quantum implication.
\newblock International journal of theoretical physics \textbf{38}(3), 793--797
  (1999)

\bibitem{saharia2019elementary}
Saharia, A., Maddila, R.K., Ali, J., Yupapin, P., Singh, G.: An elementary
  optical logic circuit for quantum computing: a review.
\newblock Optical and Quantum Electronics \textbf{51}, 1--13 (2019)

\end{thebibliography}
	
\end{document}